\documentclass[a4paper,11pt]{article}

\usepackage[margin=1in]{geometry}

\usepackage[utf8x]{inputenc}
\usepackage{amsmath}
\usepackage{amssymb}
\usepackage{amsthm}
\usepackage{algorithmic}
\usepackage{graphicx}
\usepackage{algorithm}
\usepackage{color}
\newcommand{\col}{}

\newcommand{\eps}{\varepsilon}

\newcommand{\SoS}{\text{\sc{SoS}}}

\newcommand{\set}[1]{\left\{ #1 \right\}}
\newcommand{\PS}{\mathcal{P}}

\newcommand{\y}{y^N}
\newcommand{\fallfac}[2]{{#1}^{\underline{#2}}}
\newcommand{\risefac}[2]{{#1}^{\overline{#2}}}
\newcommand{\K}{\omega}
\newcommand{\R}{\mathbb{R}}

\usepackage{mathtools}

\DeclarePairedDelimiter\floor{\lfloor}{\rfloor}

\newtheorem{theorem}{Theorem}
\newtheorem{lemma}[theorem]{Lemma}

\newtheorem{corollary}[theorem]{Corollary}
\newtheorem{definition}{Definition}

\title{Sum-of-squares hierarchy lower bounds\\ for symmetric formulations\thanks{Supported by the Swiss National Science Foundation project 200020-144491/1 ``Approximation Algorithms for Machine Scheduling Through Theory and Experiments'' and by Sciex Project 12.311.}}

\author{ Adam Kurpisz \quad Samuli Lepp\"anen~\quad Monaldo Mastrolilli\\
\small{IDSIA, 6928 Manno,
Switzerland}
{\{adam, samuli, monaldo\}@idsia.ch}
}

\date{}

\begin{document}
\maketitle
\begin{abstract}
We introduce a method for proving  Sum-of-Squares (SoS)/ Lasserre hierarchy lower bounds when the initial problem formulation exhibits a high degree of symmetry.
Our main technical theorem allows us to reduce the study of the positive semidefiniteness to the analysis of ``well-behaved'' univariate polynomial inequalities.

We illustrate the technique on two problems, one unconstrained and the other with constraints.
More precisely, we give a short elementary proof of Grigoriev/Laurent lower bound for finding the integer cut polytope of the complete graph. We also show that the SoS hierarchy requires a non-constant number of rounds to improve the initial integrality gap of 2 for the \textsc{Min-Knapsack} linear program strengthened with cover inequalities. 

\end{abstract}
\section{Introduction}
Proving lower bounds for the Sum-of-Squares (SoS)/Lasserre hierarchy~\cite{Lasserre01,parrilo00} has attracted notable attention in the theoretical computer science community during the last decade, see e.g. \cite{BhaskaraCVGZ12,Cheung07,Chla12,Grigoriev01,Grigoriev01b,Laurent03a,LeeRagSteu15,MekaPW15,Schoenebeck08,Tulsiani09}. This is partly because the hierarchy captures many of the best known approximation algorithms based on semidefinite programming (SDP) for several natural 0/1 optimization problems (see~\cite{LeeRagSteu15} for a recent result). Indeed, it can be argued that the SoS hierarchy is the strongest candidate to be the ``optimal'' meta-algorithm predicted by the Unique Games Conjecture (UGC)~\cite{Khot02a,Raghavendra08}. On the other hand, the hierarchy is also one of the best known candidates for refuting the conjecture since it is still conceivable that one could show that the SoS hierarchy achieves better approximation guarantees than the UGC predicts (see \cite{BarakS14} for discussion). Despite the interest in the algorithm and due to the many technical challenges presented by semidefinite programming, only relatively few techniques are known for proving lower bounds for the hierarchy. In particular, several integrality gap results follow from applying gadget reductions to the few known original lower bound constructions.

Indeed, many of the known lower bounds for the SoS hierarchy originated in the works of Grigoriev~\cite{Grigoriev01,Grigoriev01b}.\footnote{More precisely, Grigoriev considers the positivstellensatz proof system, which is the dual of the SoS hierarchy considered in this paper. For brevity, we will use SoS hierarchy/proof system interchangeably.}
We defer the formal definition of the hierarchy for later and only point out that solving the hierarchy after $t$ rounds takes $n^{O(t)}$ time.
In~\cite{Grigoriev01b} Grigoriev showed that random \textsc{3Xor} or \textsc{3Sat} instances cannot be solved even by $\Omega(n)$ rounds of the SoS hierarchy (some of these results were later independently rediscovered by Schoenebeck~\cite{Schoenebeck08}). Lower bounds, such as those of~\cite{BhaskaraCVGZ12,Tulsiani09} rely on \cite{Grigoriev01b,Schoenebeck08} combined with gadget reductions. Another important lower bound was also given by Grigoriev~\cite{Grigoriev01} for the \textsc{Knapsack} problem (a simplified proof can be found in~\cite{GrigorievHP02}), showing that the SoS hierarchy cannot prove within $\lfloor n/2 \rfloor$ rounds that the polytope $\{x\in [0,1]^n: \sum_{i=1}^n x_i= n/2\}$ contains no integer point when $n$ is odd. Using essentially the same construction as in~\cite{GrigorievHP02}, Laurent~\cite{Laurent03a} independently showed that $\lfloor \frac{n}{2} \rfloor$ rounds are not enough for finding the integer cut polytope of the complete graph with $n$ nodes, where $n$ is odd (this result was recently shown to be tight in~\cite{FawziSaundersonParrilo15}).\footnote{The two problems, \textsc{Knapsack} and \textsc{Max-Cut} in complete graphs, considered respectively in~\cite{Grigoriev01,GrigorievHP02} and in~\cite{Laurent03a}, are essentially the same and we will use \textsc{Max-Cut} to refer to both.} By using several new ideas and techniques, but a similar starting point as in~\cite{GrigorievHP02,Laurent03a}, Meka, Potechin and Wigderson~\cite{MekaPW15} were able to show a lower bound of $\Omega(\log^{1/2} n)$ for the \textsc{Planted-Clique} problem. Common to the works~\cite{GrigorievHP02,Laurent03a} and~\cite{MekaPW15} is that the matrix involved in the analysis has a large kernel, and they prove that a principal submatrix is positive definite by applying the theory of association schemes \cite{Godsil10}.
It is also interesting to point out that for the class of \textsc{Max-CSP}s, Lee, Raghavendra and Steurer~\cite{LeeRagSteu15} proved that the SoS relaxation yields the ``optimal'' approximation, meaning that SDPs of polynomial-size are equivalent in power to those
arising from $O(1)$ rounds of the SoS relaxations. Then, by appealing to the result by Grigoriev/Laurent~\cite{Grigoriev01,Laurent03a} they showed an exponential lower bound on the size of SDP formulations for the integer cut polytope.
For different techniques to obtain lower bounds, we refer for example to the recent papers \cite{BarakCK15,KurpiszLM15b,KurpiszLM15} (see also Section~\ref{sect:further}) and the survey \cite{Chla12} for an overview of previous results.

In this paper we introduce a method for proving  SoS hierarchy lower bounds when the initial problem formulation exhibits a high degree of symmetry. Our main technical theorem (Theorem \ref{thm:PSD_as_polynomial}) allows us to reduce the study of the positive semidefiniteness to the analysis of ``well-behaved'' univariate polynomial inequalities. The theorem applies whenever the solution and constraints are \emph{symmetric}, informally meaning that all subsets of the variables of equal cardinality play the same role in the formulation (see Section \ref{sect:main_theorem} for the formal definition). For example, the solution in \cite{Grigoriev01,GrigorievHP02,Laurent03a} for \textsc{Max-Cut} is symmetric in this sense. 

We note that exploiting symmetry reduces the number of variables involved in the analysis, and different ways of utilizing symmetry have been widely used in the past for proving integrality gaps for different hierarchies, see for example~\cite{BlekhermanGP14,GoemansT01,Grigoriev01b,HongT08,KurpiszLM15,StephenT99}.
{\col An interesting difference of our approach from others is that we establish several lower bounds without fully identifying the formula of eigenvectors. 
More specifically, the common task in this context is to identify the spectral structure to get a simple diagonalized form. In the previous papers the moment matrices belong to the Bose-Mesner algebra of a well-studied association scheme, and hence one can use the existing theory. In this paper, instead of identifying the spectral structure completely, we identify only possible forms and propose to test all the possible candidates.
This is in fact an important point, since the approach may be extended even if the underlying symmetry is imperfect or its spectral property is not well understood.
}

The proof of Theorem \ref{thm:PSD_as_polynomial} is obtained by a sequence of elementary operations, as opposed to notions such as big kernel in the matrix form, the use of interlacing eigenvalues, the machinery of association schemes and various results about hyper-geometric series as in~\cite{Grigoriev01,GrigorievHP02,Laurent03a}.  Thus Theorem~\ref{thm:PSD_as_polynomial} applies to the whole class of symmetric solutions, even when several conditions and machinery exploited in \cite{Grigoriev01,GrigorievHP02,Laurent03a} cannot be directly applied. For example the kernel dimension, which was one of the important key property used to prove the results in ~\cite{Grigoriev01,GrigorievHP02,Laurent03a}, depends on the particular solution that is used and it is not a general property of the class of symmetric solutions. The solutions for two problems considered in this paper have completely different kernel sizes of the analyzed matrices, one large and the other zero.

We demonstrate the technique with two illustrative and complementary applications. First, we show that the analysis of the lower bound for \textsc{Max-Cut} in \cite{Grigoriev01,GrigorievHP02,Laurent03a} simplifies to few elementary calculations once the main theorem is in place. This result is partially motivated by the open question posed by O'Donnell~\cite{ODonnell13} of finding a simpler proof for Grigoriev's lower bound for the \textsc{Knapsack} problem. 


As a second application we consider a constrained problem. We show that after $\Omega(\log^{1-\epsilon} n)$ levels the SoS hierarchy does not improve the integrality gap of $2$ for the \textsc{Min-knapsack} linear program formulation strengthened with \emph{cover inequalities}~\cite{CarrFLP00} introduced by Wolsey \cite{Wolsey75}. Adding cover inequalities is currently the most successful approach for capacitated covering problems of this type~\cite{BansalBN08,BansalGK10,BansalP10,CarnesS08,ChakrabartyGK10}.

{\col Our result is the first SoS lower bound for formulations with cover inequalities. In this application we demonstrate that our technique can also be used for suggesting the solution and for analyzing its feasibility.

Finally we point it out that the same analysis can be used to provide a non trivial lower bound to an open question raised by Laurent \cite{Laurent03} regarding the Lasserre rank of the knapsack problem (see Section \ref{sect:further} for a discussion).
}



\section{The SoS hierarchy}

Consider a $0/1$ optimization problem with $m \geq 0$ linear constraints $g_{\ell}(x)\geq 0$, for $\ell\in[m]$ and $x \in \mathbb{R}^n$. We are interested in approximating the convex hull of the integral points of the set $K = \set{x \in \mathbb{R}^n~|~g_{\ell}(x)\geq 0, \forall \ell\in [m]}$ with the SoS hierarchy defined in the following.

The form of the SoS hierarchy we use in this paper (Definition~\ref{def:sos_definition}) is equivalent to the one used in literature (see e.g.~\cite{BarakBHKSZ12,Lasserre01,Laurent03}). It follows from applying a change of basis to the dual certificate of the refutation of the proof system~\cite{Laurent03} (see also \cite{MekaPW15} for discussion on the connection to the proof system). We use this change of basis  in order to obtain a useful decomposition of the moment matrices as a sum of rank one matrices of special kind. This will play an important role in our analysis. We refer the reader to Appendix \ref{app:SoS} for more details and for a mapping between the different forms.

For any $I\subseteq N=\{1,\ldots,n\}$, let $x_I$ denote the $0/1$ solution obtained by setting $x_i = 1$ for $i \in I$, and $x_i = 0$ for $i\in N\setminus I$. We denote by $g_\ell(x_I)$ the value of the constraint evaluated at $x_I$.
For each integral solution $x_I$, where $I\subseteq N$, in the SoS hierarchy defined below there is a variable $\y_I$ that can be interpreted as the ``relaxed'' indicator variable for the solution $x_I$.
We point out that in this formulation of the hierarchy the number of variables $\{\y_I:I\subseteq N\}$ is exponential in $n$, but this is not a problem in our context since we are interested in proving lower bounds rather than solving an optimization problem.

Let $\PS_t(N)$ be the collection of subsets of $N$ of size at most $t\in\mathbb{N}$. For every $I \subseteq N$, the $q$-\text{zeta vector}  $Z_I \in \mathbb{R}^{\PS_q(N)}$ is a $0/1$ vector with $J$-th entry ($|J| \leq q$) equal to $1$ if and only if $J \subseteq I$.\footnote{In order to keep the notation simple, we do not emphasize the parameter $q$ as the dimension of the vectors should be clear from the context.}
Note that $Z_IZ_I^\top$ is a rank one matrix and the matrices considered in Definition~\ref{def:sos_definition} are linear combinations of these rank one matrices.

\begin{definition}~\label{def:sos_definition}
 The $t$-th round SoS hierarchy relaxation for the set $K$, denoted by $\SoS_t(K)$, is the set of values $\{\y_I \in \mathbb{R}:  \forall I \subseteq N\}$ that satisfy
 \begin{eqnarray}
 \sum_{\substack{I \subseteq N}} \y_I&=& 1, \label{eq:sos_sum_condition} \\
 \sum_{\substack{I \subseteq N}} \y_I Z_I Z_I^\top &\succeq& 0, \text{ where } Z_I \in \mathbb{R}^{\PS_{t+d}(N)} \label{eq:sos_variables}\\
 \sum_{\substack{I \subseteq N}} g_\ell(x_I)\y_I Z_IZ_I^\top &\succeq& 0, ~ \forall \ell\in [m]  \text{, where }  Z_I \in \mathbb{R}^{\PS_{t}(N)}  \label{eq:sos_constraints}
\end{eqnarray}
where $d = 0$ if $m = 0$ (no linear constraints), otherwise $d = 1$.
\end{definition}
It is straightforward to see that the SoS hierarchy formulation given in Definition~\ref{def:sos_definition} is a relaxation of the integral polytope. Indeed consider any feasible integral solution $x_I \in K$ and set $\y_I=1$ and the other variables to zero. This solution clearly satisfies Condition~\eqref{eq:sos_sum_condition}, Condition~\eqref{eq:sos_variables} because the rank one matrix $Z_IZ_I^\top$ is positive semidefinite (PSD), and Condition~\eqref{eq:sos_constraints} since $x_I\in K$.

\section{The main technical theorem} \label{sect:main_theorem}

{\col The main result of this paper (see Theorem~\ref{thm:PSD_as_polynomial} below) allows us to reduce the study of the positive semidefiniteness for matrices~\eqref{eq:sos_variables} and~\eqref{eq:sos_constraints} to the analysis of ``well-behaved'' univariate polynomial inequalities.} It can be applied whenever the solutions and constraints are \emph{symmetric}, namely they are invariant under all permutations $\pi$ of the set $N$: $z_I^N=z^N_{\pi(I)}$ for all $I\subseteq N$ (equivalently when $z_I^N = z^N_J$ whenever $|I|=|J|$),\footnote{We define the set-valued permutation by $\pi(I) = \set{\pi(i)~|~i \in I}$. \label{footnote:set_permutation}} where $z_I^N$ is understood to denote either $y^N_I$ or $g_\ell(x_I)y^N_I$.
For example, the solution for \textsc{Max-Cut} considered by Grigoriev~\cite{Grigoriev01} and Laurent~\cite{Laurent03a} belongs to this class. 

{\col
\begin{theorem} \label{thm:PSD_as_polynomial}
For any $t\in \{1,\ldots,n\}$, let $\mathcal{S}_t$ be the set of all polynomials $G_h(k)\in \R[k]$, for $h \in \{0,\ldots,t\}$, that satisfy the following conditions: 
\begin{align}
G_h(k)&\in \R[k]_{2t} \label{eq:degree}\\
G_h(k)&=0 \qquad \text{for }k\in \{0,\ldots,h-1\}\cup \{n-h+1,\ldots,n\}\label{eq:zeros}\\
G_h(k)&\geq0 \qquad \text{for }k\in [h-1,\ldots,n-h+1]\label{eq:geq0}
\end{align} 
For any fixed set of values $\{z^N_k \in \R:k =0,
\ldots,n\}$,
if the following holds 
\begin{align}
\sum_{k=h}^{n-h} z^N_k \binom{n}{k} G_h(k) &\geq 0 \qquad \forall G_h(k)\in \mathcal{S}_t \label{eq:sym_psd_cond}
\end{align}
then matrix \eqref{matrixM} is positive-semidefinite 
\begin{equation}\label{matrixM}
 \sum_{k = 0}^n z^N_k \sum_{\substack{I \subseteq N \\ |I| = k}} Z_IZ_I^\top \qquad (\text{where } Z_I \in \mathbb{R}^{\PS_{t}(N)})
\end{equation}
\end{theorem}
Note that polynomial $G_h(k)$ in~\eqref{eq:geq0} is nonnegative in a \emph{real interval}, and in~\eqref{eq:zeros} it is zero for a \emph{finite set of integers}. Moreover, constraints \eqref{eq:sym_psd_cond} are trivially satisfied for $h> \floor{n/2}$.

Theorem~\ref{thm:PSD_as_polynomial} is a actually a corollary of a technical theorem that is not strictly necessary for the applications of this paper, and therefore deferred to a later section (see Theorem~\ref{thm:PSD_as_polynomial_full} in Section~\ref{sect:th1}).
The proof (given in Section~\ref{sect:th1}) is obtained by exploiting the high symmetry of the eigenvectors of the matrix appearing in \eqref{matrixM}. Condition~\eqref{eq:sym_psd_cond} corresponds to the requirement that the Rayleigh quotient being non-negative restricted to some highly symmetric vectors (which we show are the only ones we need to consider). }
%

\section{Max-Cut for the complete graph} \label{sect:maxcut}
In the \textsc{Max-Cut} problem, we are given an undirected graph
and we wish to find a partition of the vertices (a cut) which maximizes the number of edges whose
endpoints are on different sides of the partition (cut value).
For the complete graph with $n$ vertices, consider any solution with $\K$ vertices on one side and the remaining $n-\K$ on the other side of the partition. This gives a cut of value $\K(n-\K)$.
%
When $n$ is odd and for \emph{any} $\K\leq n/2$, Grigoriev~\cite{Grigoriev01} and Laurent~\cite{Laurent03a} considered the following solution (reformulated in the basis considered in Definition~\ref{def:sos_definition}, see Appendix~\ref{Sect:app_max_cut}):
 \begin{equation}\label{prob11}
\y_I= (n+1) \binom{{\K}}{n+1} \frac{(-1)^{n-|I|}}{{\K}-|I|} \qquad \forall I\subseteq N
\end{equation}
It is shown~\cite{Grigoriev01,Laurent03a} that~\eqref{prob11}  is a feasible solution for the SoS hierarchy of value $\K(n-\K)$, for \emph{any} $\K\leq n/2$ up to round $t\leq \lfloor \K\rfloor$. In particular for $\K=n/2$ the cut value of the SoS relaxation is strictly larger than the value of the optimal integral cut (i.e. $\lfloor\frac{n}{2}\rfloor(\lfloor\frac{n}{2}\rfloor+1)$), showing therefore an integrality gap at round~$\lfloor n/2\rfloor$.

We note that the formula for the solution~\eqref{prob11} is essentially implied by the requirement of having exactly $\K$ vertices on one side of the partition (see~\cite{Grigoriev01,Laurent03a} and~\cite{MekaPW15} for more details) and the core of the analysis in~\cite{Grigoriev01,Laurent03a} is in showing that~\eqref{prob11} is a feasible solution for the SoS hierarchy.
By taking advantage of Theorem~\ref{thm:PSD_as_polynomial}, the proof that~\eqref{prob11} is a feasible solution for the SoS relaxation follows by observing the fact below.
%
\begin{lemma}\label{th:polypreserve}
For any polynomial $P(x)\in\mathbb{R}[x]$ of degree $\leq n$ and $y_{i}^N=\y_I$ defined in~\eqref{prob11} we have
$$
\sum_{k=0}^n \binom{n}{k} y_k^N \,  P(k) = P(\K)
$$
\end{lemma}
\begin{proof}
By the polynomial remainder theorem $P(k)=(\K-k)Q(k)+P(\K)$, where $Q(k)$ is a unique polynomial of degree at most $n-1$. It follows that
\begin{align}
\sum_{k=0}^n \binom{n}{k} y_{k}^N \,  P(k) &= \underbrace{\sum_{k=0}^n \binom{n}{k} y_{k}^N \,  (\K-k)Q(k)}_{=0 }+P(\K) \underbrace{\sum_{k=0}^n \binom{n}{k} y_{k}^N}_{=1}=P(\K) \notag
\end{align}
since $\sum_{k=0}^n (-1)^k\binom{n}{k} Q(k)=0$ for any polynomial of degree at most $n-1$.

\end{proof}
Now by Lemma~\ref{th:polypreserve} we have $\sum_{k=0}^n y^N_k \binom{n}{k} G_h(k) =G_h(\K)$ and the feasibility of~\eqref{prob11} follows by Theorem~\ref{thm:PSD_as_polynomial}, since we have that $G_h(\omega) \geq 0$ whenever $t \leq \omega$ for $\omega \leq n/2$.


\section{Min-Knapsack with cover inequalities}\label{sect:knapgap}

The \textsc{Min-Knapsack} problem is defined as follows: we have $n$ items with costs $c_i$ and profits $p_i$, and we want to choose a subset of items such that the sum of the costs of the selected items is minimized and the sum of the profits is at least a given demand $P$. Formally, this can be formulated as an integer program $(IP) ~ \min \left\{ \sum_{j =1}^{n} c_jx_j : \\ \sum_{j =1}^{n} p_jx_j \geq P, x \in \{0,1\}^n \right\}$. It is easy to see that the natural linear program $(LP)$, obtained by relaxing $x \in \{0, 1\}^n$ to $x \in [0, 1]^n$ in $(IP)$, has an unbounded integrality gap.

By adding the \emph{Knapsack Cover} (KC) inequalities introduced by Wolsey \cite{Wolsey75} (see also \cite{CarrFLP00}), the arbitrarily large integrality gap of the natural LP can be reduced to 2 (and it is tight~\cite{CarrFLP00}).
The  KC  constraints are as follows:
$ \sum_{j \not \in A} p^A_j x_j \geq P- p(A)$ for all $A \subseteq N$,
where $p(A)=\sum_{i\in A}p_i$ and $p^A_j = \min \set{p_j, P- p(A)}$. Note that these constraints are valid constraints for integral solutions. Indeed, in the ``integral world'' if a set $A$ of items is picked we still need to cover $P-p(A)$; the remaining profits are ``trimmed'' to be at most $P-p(A)$ and this again does not remove any feasible integral solution.

The following instance~\cite{CarrFLP00} shows that the integrality gap implied by  KC  inequalities is~2: we have $n$ items of unit costs and profits. We are asked to select a set of items in order to obtain a profit of at least $1+1/(n-1)$. The resulting linear program formulation with  KC  inequalities is as follows (for  $x_i \in [0,1]$, $i=1,\ldots,n$)
\begin{eqnarray}
 \mbox{($LP^+$)}
  \min  \sum_{j =1}^{n} x_j   & \quad \mbox{s.t.} &  \sum_{j =1}^{n} x_j \geq 1+1/(n-1) \label{eq:knapsack_constraint}\\
 &  &  \sum_{j \in N'} x_j \geq 1 \qquad \forall N'\subseteq N:|N'|=n-1 \label{eq:knapsack_covering_constraint}
\end{eqnarray}
Note that the solution $x_i=1/(n-1)$ is a valid fractional solution of value $1+1/(n-1)$ whereas the optimal integral solution has value $2$.
In the following we show that $\SoS_t(LP^+)$, with $t$ arbitrarily close to {\col a logarithmic function of $n$}, admits the same integrality gap as the initial linear program ($LP^+$) relaxation.
\begin{theorem}
\label{thm:knapsack_gap_main}
For any $\delta >0$ and sufficiently large $n'$, let $t = \lfloor \log^{1-\delta} n'\rfloor$, $n=\lfloor \frac{n'}{t} \rfloor t$ and $\epsilon =o(t^{-1})$. Then the following solution is feasible for $\SoS_t(LP^+)$ with integrality gap of $2- o(1)$
\begin{eqnarray}\label{eq:knapsack_gap_solution}
\y_I= \binom{n}{|I|}^{-1} \cdot
\left\{ \begin{array}{ll}
\frac{(1 + \epsilon)n}{(n-1)\lfloor \log n\rfloor} & \text{for }|I|=\lfloor \log n\rfloor\\
\frac{\epsilon t}{jn}  & \text{for }|I|= j\frac{n}{t} \text{ and } j\in[t]\\
1-\sum_{\emptyset \neq I \subseteq N} \y_I & \text{for }I=\emptyset\\
0 & \text{otherwise}
\end{array}
\right.
\end{eqnarray}
\end{theorem}

\subsection{Overview of the proof}

An integrality gap proof for the SoS hierarchy can in general be thought of having two steps: first, choosing a solution to the hierarchy that attains a superoptimal value, and second showing that this solution is feasible for the hierarchy. We take advantage of Theorem~\ref{thm:PSD_as_polynomial} in both steps. Here we describe the overview of our integrality gap construction while keeping the discussion informal and technical details minimal, the proof can be found in Appendix~\ref{sect:proof_of_knapsack_cover_thm}.

\paragraph{Choosing the solution.}
We make the following simplifying assumptions about the structure of the variables $y_I^N$: due to symmetry in the problem we set $y_I^N = y_J^N = y_k^N$ for each $I, J$ such that $|I| = |J|=k$, and for every $I \subseteq N$ we set $y_I^N \geq 0$ in order to satisfy~\eqref{eq:sos_variables} for free. Furthermore, in order to have an integrality gap (i.e. a small objective function value), we guess that $y^N_0 \approx 1$ forcing the other variables to be small due to~\eqref{eq:sos_sum_condition}.

We then show that satisfying~\eqref{eq:sos_constraints} for every constraint follows from showing that
\begin{equation} \label{eq:choosing_the_solution_psdness}
\sum_{k=0}^n \binom{n}{k} y_k^N(k-2)\prod_{i=1}^t(k-r_i)^2 \geq 0
\end{equation}
for every choice of $t$ real variables $r_i$. We get this condition by observing similarities in the structure of the constraints and applying Theorem~\ref{thm:PSD_as_polynomial}, then expressing the polynomial in root form.\footnote{We show that the roots $r_i$ can be assumed to be real numbers.} If we set $y_1^N = 0$, the only negative term in the sum corresponds to $y_0^N$. Then, it is clear that we need at least $t+1$ non-zero variables $y_k^N$, otherwise the roots $r_i$ can be set such that the positive terms in~\eqref{eq:choosing_the_solution_psdness} vanish and the inequality is not satisfied. Therefore, we choose exactly $t+1$ of the $y_k^N$ to be strictly positive (and the rest 0 excluding $y_0^N$), and we distribute them ``far away'' from each other, so that no root can be placed such that the coefficient of two positive terms become small. To take this idea further, for one ``very small'' $k'$ (logarithmic in $n$), we set $y_{k'}^N$ positive and space out the rest evenly.

\paragraph{Proving that the solution is feasible.}

We show that~\eqref{eq:choosing_the_solution_psdness} holds for all possible $r_i$ with our chosen solution by analysing two cases. In the first case we assume that all of the roots $r_i$ are larger than $\log^3 n$. Then, we show that the ``small'' point $k'$ we chose is enough to satisfy the condition. In the complement case, we assume that there is at least one root $r_i$ that is smaller than $\log^3 n$. It follows that one of the evenly spaced points is ``far'' from any remaining root, and can be used to show that the condition is satisfied.


\subsection{Proof of Theorem~\ref{thm:knapsack_gap_main}} \label{sect:proof_of_knapsack_cover_thm}

We start by proving the claimed integrality gap. The defined solution has an objective value that is arbitrarily close to 1 whereas the optimal integral value is 2. Indeed, the objective value of the relaxation is (see Appendix \ref{app:SoS}):
$
\sum_{I\subseteq N} \y_I|I| = \frac{n}{n-1}(1+\varepsilon) + \varepsilon t  \stackrel{n\rightarrow \infty}{=} 1.
$


The remaining part of the Theorem~\ref{thm:knapsack_gap_main} follows by showing that the suggested solution satisfies~\eqref{eq:sos_sum_condition},~\eqref{eq:sos_variables} and~\eqref{eq:sos_constraints}. Note that~\eqref{eq:sos_sum_condition} is immediately satisfied by the definition of variables $\{y_I^N\}$, and~\eqref{eq:sos_variables} is satisfied since $y_I^N \geq 0$ and the rank one matrix $Z_I Z_I^{\top}$ is positive semidefinite for every $I \subseteq N$. It remains to prove that the condition~\eqref{eq:sos_constraints} is also satisfied for all the constraints~\eqref{eq:knapsack_constraint} and~\eqref{eq:knapsack_covering_constraint}. Note that the constraint~\eqref{eq:knapsack_covering_constraint} is not symmetric (one variable is missing and sets of variables of the same size do not play the same role with respect to this constraint). However, the following lemma shows how to solve this issue by reducing to the form~\eqref{matrixM} of Theorem \ref{thm:PSD_as_polynomial}.

\begin{lemma} \label{lemma:expression}
The constraint~\eqref{eq:sos_constraints} holds for both~\eqref{eq:knapsack_constraint} and~\eqref{eq:knapsack_covering_constraint} if the solution of Theorem~\ref{thm:knapsack_gap_main} satisfies
\begin{equation} \label{eq:sufficient_condition_KC}
\sum_{\emptyset \neq I \subseteq N}^{n}y_I^N \left(|I|-2\right) Z_I Z_I^{\top} \succeq \frac{n}{n-1} Z_{\emptyset} Z_{\emptyset}^{\top}
\end{equation}
\end{lemma}
\begin{proof}
 We first show that this implies that~\eqref{eq:sos_constraints} holds for the constraint $\sum_{j =1}^{n} x_j \geq 1+1/(n-1)$.
Since for a large $n$, $y_I^N=0$ for $|I|=1$ and $y_{\emptyset}^N \leq 1$, the condition~\eqref{eq:sos_constraints} takes the following form
$$
\sum_{\emptyset \neq I  \subseteq N} y_I^N\left(|I|-\frac{n}{n-1}\right)Z_I Z_I^{\top} \succeq y_{\emptyset}^N\frac{n}{n-1} Z_{\emptyset} Z_{\emptyset}^{\top}
$$
which is implied if (\ref{eq:sufficient_condition_KC}) is satisfied (recall that $\y_I\geq 0$).
Next, we will show that~\eqref{eq:sufficient_condition_KC} also implies that~\eqref{eq:sos_constraints} is satisfied for the cover constraint $\sum_{j = 1}^{n-1} x_j \geq 1$ (the other cases are similar). For this constraint Condition~\eqref{eq:sos_constraints} can be written as
\begin{eqnarray*}
\sum_{\substack{n\notin I \subseteq N\\ I \neq \emptyset}} \y_I(|I|-1) Z_IZ_I^\top + \sum_{\substack{n\in I \subseteq N}} \y_I(|I|-2) Z_IZ_I^\top   - \y_{\emptyset} Z_{\emptyset}Z_{\emptyset}^\top  \\
\succeq  \sum_{\substack{I \subseteq N\\I \neq \emptyset}} \y_I(|I|-2) Z_IZ_I^\top  - \y_{\emptyset} Z_{\emptyset}Z_{\emptyset}^\top \succeq 0
\end{eqnarray*}
which is also implied if (\ref{eq:sufficient_condition_KC}) is satisfied.

\end{proof}

Now by Theorem~\ref{thm:PSD_as_polynomial}, Condition~\eqref{eq:sufficient_condition_KC}
holds if we have
$
\sum_{k=1}^n y_k^N (k-2) \binom{n}{k}G_h(k) \geq  \frac{n}{n-1} G_h(0)
$
for $h=0,1,\ldots,t$ and every univariate polynomial $G_h(k)$ of degree $2t$ such that $G_h(k) \geq 0$ for $k  \in [h-1,\ldots,n-h+1]$ and $G_h(k)=0$ for $k  \in \{0,\ldots,h-1\}\cup\{n-h+1,\ldots,n\}$.

Note that the only nontrivial case is for $h=0$, since otherwise the above condition is immediately satisfied. Indeed, for $h > 0$, we have that $G_h(0) = 0$ and the only remaining terms in the sum are non-negative. Thus in order to complete the proof of Theorem~\ref{thm:knapsack_gap_main} it is enough to show that the following is satisfied
\begin{equation}
\sum_{k=1}^n y_k^N (k-2) \binom{n}{k}P^2(k) \geq  \frac{n}{n-1} P^2(0) \quad \forall P: \text{deg}(P)\leq t
\label{eq:knapsack_condition_in_poly_form}
\end{equation}
The following lemma (proved in Section~\ref{sect:proof_of_lemma_three_facts}) further reduces the interesting cases.
\begin{lemma} \label{lemma:three_facts}
In order to prove that Solution~\eqref{eq:knapsack_gap_solution} satisfies (\ref{eq:knapsack_condition_in_poly_form}) it is sufficient to prove that \eqref{eq:knapsack_gap_solution} satisfies (\ref{eq:knapsack_condition_in_poly_form}) for polynomials $P(x)$ with the following properties:
\begin{enumerate}
 \item[(a)] all the roots $r_1,..., r_t$ of $P(x)=0$ are real numbers,
 \item[(b)] all the roots $r_1,..., r_t$ of $P(x)=0$  are in the range, $1 \leq r_j \leq n$ for all $j = 1,\ldots,t$,
 \item[(c)] the degree of $P(x)$ is exactly $t$.
\end{enumerate}
\end{lemma}

Next we show that Solution~\eqref{eq:knapsack_gap_solution} satisfies (\ref{eq:knapsack_condition_in_poly_form}) and that there exists an $\varepsilon = o(t^{-1})$ as claimed.

The fundamental theorem of algebra states that any univariate polynomial of degree $t$ has exactly $t$ complex roots. According to Lemma~\ref{lemma:three_facts} we can focus to polynomial with $t$ real roots. We prove that the suggested solution satisfies (\ref{eq:knapsack_condition_in_poly_form}) by expressing the generic univariate polynomial $P(k)$ using its roots $r_1, ..., r_t$, so that (\ref{eq:knapsack_condition_in_poly_form}) becomes

\begin{equation}\label{eq:silp}
\sum_{k = 1}^n \binom{n}{k} y_k^N\left(k-2 \right) \prod_{i=1}^t (r_i-k)^2  \geq  \left(1+\frac{1}{n-1} \right) \prod_{i=1}^t r_i^2
\end{equation}

To show that (\ref{eq:silp}) is satisfied we separate two cases: when all of the roots of the polynomial are greater or equal to a fixed treshold $\alpha = \log^3 n$ and when at least one root is smaller than this treshold. In order to simplify the computations we denote $\beta = \lfloor \log n \rfloor$.
\begin{enumerate}
 \item $r_j \geq \alpha$ for all $j$. It is sufficient to show that the left--hand side term in (\ref{eq:silp}) corresponding to $k = \beta$ satisfies
$$
\binom{n}{\beta} y_{\beta}^N\left(\beta-2\right)\prod_{i=1}^t (r_i-\beta)^2 \geq \left(1+ \frac{1}{n-1} \right) \prod_{i=1}^t r_i^2
$$
Replacing the variables with the values we get
\begin{align*}
\frac{n}{n-1}\frac{1+\varepsilon}{\beta}\left(\beta-2\right)\prod_{i=1}^t (r_i-\beta)^2  \geq   \frac{n}{n-1} \prod_{i=1}^t r_i^2 \\
\Longleftrightarrow 1+\varepsilon \geq \prod_{i=1}^t\left(\frac{r_i}{r_i-\beta}\right)^2 \frac{1}{1-2\beta^{-1}}
\end{align*}
Since by Theorem~\ref{thm:PSD_as_polynomial} and assumption, all roots $r_j$ satisfy $\alpha \leq r_j \leq n$. Since $\frac{r_i}{r_i - \beta} \leq \frac{\alpha}{\alpha-\beta}$ it is sufficient that it holds
$
1 + \varepsilon \geq \frac{1}{1-2\beta^{-1}} \left(\frac{\alpha}{\alpha-\beta}\right)^{2t}
$.

 \item There is at least one root $r_j$ such that $r_j < \alpha$. It can be shown by straightforward induction on the number of roots that if for at least one $j$, $r_j < \alpha$, then there exists a point $u = l\frac{n}{t}$, $l = 1,...,t$ such that $\binom{n}{u} y_u> 0$ and
$
|u - r_i| \geq \frac{n}{2t}
$
for all $i = 1,...,t$. Let $u$ be such a point. It is sufficient to show that we can satisfy
$
\binom{n}{u} y_u^N\left(u-2\right)\prod_{i=1}^t (r_i-u)^2 \geq \frac{n}{n-1}\prod_{i=1}^t r_i^2
$.

We have
$
\binom{n}{u} y_u^N = \frac{\varepsilon}{u}
$
and the estimates
$
u-2 \geq \frac{u}{2}, 
(r_i-u)^2 \geq \frac{n^2}{(2t)^2}, 
\prod_{i=1}^t r_i \leq n^{t-1}\alpha
$.
Substituting these we get the condition
$
\frac{\varepsilon}{2}\left(\frac{n}{2t}\right)^{2t} \geq \frac{n}{n-1} n^{2t-2}\alpha^2
$
which gives us the requirement that
$
\varepsilon \geq  \frac{2 \alpha^2}{n^2} (2t)^{2t} \frac{n}{n-1}
$.
\end{enumerate}
These two cases suggest that we fix $\varepsilon$ as
$$
\varepsilon = \max\left\{\frac{1}{1-2\beta^{-1}}\left( 1-\frac{\beta}{\alpha}\right)^{-2t}-1,  \frac{n}{n-1}\frac{2 \alpha^2}{n^2} (2t)^{2t} \right\}
$$
The proof has now been reduced to showing that with this choice of $\varepsilon$ we have $\varepsilon t \rightarrow 0$, i.e., $\varepsilon = o(t^{-1})$.
Assume $\varepsilon = \frac{1}{1-2\beta^{-1}}\left( 1-\frac{\beta}{\alpha}\right)^{-2t}-1$. Then
$
\begin{array}{l}
\varepsilon t = t\left(\frac{1}{1-2\beta^{-1}}\left( 1-\frac{\beta}{\alpha}\right)^{-2t}-1\right) \leq t\left(\frac{1}{1-2\beta^{-1}}e^{4t\frac{\beta}{\alpha}}-1\right) \\
\end{array}
$,
when $\beta/\alpha \leq 1/2$, using the estimate $1-x \geq e^{-2x} \Rightarrow (1-x)^{-2t} \leq e^{4xt}$ which holds when $x \leq 1/2$. Furthermore, the same estimate yields $e^x - 1 \leq 2x$ when $x\leq 1/2$. Hence, we have the bound
$$
\varepsilon t \leq t \frac{1}{1-2\beta^{-1}} \cdot 8t\frac{\beta}{\alpha} + t \left(\frac{1}{1-2\beta^{-1}}-1\right) = \frac{8}{1-2\beta^{-1}} \cdot t^2\frac{\beta}{\alpha} + \frac{2 t\beta^{-1}}{1-2\beta^{-1}}
$$
The right-hand side goes to 0 if
$
\frac{t^2\beta}{\alpha} \rightarrow 0 \textrm{ and } \frac{t}{\beta} \rightarrow 0
$
as $n\rightarrow \infty$. This is clearly the case for $t \leq \log^{1-\delta} n$, for any $\delta >0$.

Next, assume $\varepsilon =  \frac{n}{n-1} \frac{2\alpha^2}{n^2} (2t)^{2t}$. Then
$
\varepsilon t = t \frac{n}{n-1} \frac{2\alpha^2}{n^2} (2t)^{2t},
$
which immediately yields the condition on $\alpha$ and $t$ that we need
$
\frac{t\alpha^2}{n^2} (2t)^{2t} \rightarrow 0
$
as $n \rightarrow \infty$. Substituting $\alpha = \log^3 n$ and $t = \log^{1-\delta}n$, for any $\delta >0$, allows us to write this as
$
\frac{t\alpha^2}{n^2} (2t)^{2t} = \log^{1-\delta}n\frac{\log^6 n}{n^2}(2\log^{1-\delta}n)^{2\log^{1-\delta}n}
$.
By a change of variables of the form $w = \log^{1-\delta}n$ we get
$$
\frac{w^{2w + \frac{7-\delta}{1-\delta}}2^{2w}}{e^{2w^{\frac{1}{1-\delta}}}} \leq \frac{w^{4w+ \frac{7-\delta}{1-\delta}}}{e^{2w^{\frac{1}{1-\delta}}}} = \frac{e^{(4w+ \frac{7-\delta}{1-\delta})\log w}}{e^{2w^{\frac{1}{1-\delta}}}} = e^{(4w+ \frac{7-\delta}{1-\delta})\log w-{2w^{\frac{1}{1-\delta}}}}
$$
which tends to 0 as $n \rightarrow \infty$.

\subsection{Proof of Lemma~\ref{lemma:three_facts}} \label{sect:proof_of_lemma_three_facts}
\textbf{Lemma \ref{lemma:three_facts}.}\textit{
In order to prove that Solution~\eqref{eq:knapsack_gap_solution} satisfies (\ref{eq:knapsack_condition_in_poly_form}) it is sufficient to prove that \eqref{eq:knapsack_gap_solution} satisfies (\ref{eq:knapsack_condition_in_poly_form}) for polynomials $P(x)$ with the following properties:
\begin{enumerate}
 \item[(a)] all the roots $r_1,..., r_t$ of $P(x)=0$ are real numbers,
 \item[(b)] all the roots $r_1,..., r_t$ of $P(x)=0$  are in the range, $1 \leq r_j \leq n$ for all $j = 1,\ldots,t$,
 \item[(c)] the degree of $P(x)$ is exactly $t$.
\end{enumerate}
}
\begin{proof}

First notice that (\ref{eq:knapsack_condition_in_poly_form}) is equivalent to
\begin{equation}
\label{eq:sufficient_condition_KC_root_form}
\sum_{k = 1}^n y_k^N \binom{n}{k} \left(k-2 \right) \prod_{j=1}^t \left(\frac{r_j-k}{r_j}\right)^2  \geq  \left(1+\frac{1}{n-1} \right)
\end{equation}
where for the fixed $n$ the right-hand side is constant.
\begin{enumerate}
 \item [(a)] Let $P(k)$ be the univariate polynomial with $2q$ complex roots (complex roots appear in conjugate pairs) i.e. $r_{2j-1} = a_j + b_ji$, $r_{2j} = a_j-b_ji$ for $j = 1,...,q$ and the rest real roots. Let $P'(k)$ be the polynomial with all real roots such that $r'_{2j-1} = r'_{2j} = \sqrt{a_{2j}^2 + b_{2j}^2}$ for $j=1,...,q$ and $r'_j = r_j, j>2q$.

For any $k\in N$ and $j\in[t]$, a simple calculation shows that
$$
\left(\frac{r_{2j-1}-k}{r_{2j-1}}\right)^2\left(\frac{r_{2j}-k}{r_{2j}}\right)^2 \geq \left(\frac{r'_{2j-1}-k}{r'_{2j-1}}\right)^2\left(\frac{r'_{2j}-k}{r'_{2j}}\right)^2
$$
Hence,

$$
\sum_{k = 1}^n  y_k^N \binom{n}{k}\left(k-2 \right)   \prod_{j=1}^t \left(\frac{r_j-k}{r_j}\right)^2 \geq
\sum_{k = 1}^n   y_k^N \binom{n}{k} \left(k-2 \right)  \prod_{j=1}^t \left(\frac{r'_j-k}{r'_j}\right)^2\\
$$

\item [(b)] Let $P(k)$ be the univariate polynomial with all positive roots but one i.e. $r_1 = -a$, for $a>0$. Let $P'(k)$ be the univariate polynomial with all positive roots such that $r'_{1} = a$ and $r'_j = r_j, j>1$.
Since for any $k\in N$
$$
 \left(\frac{-a-k}{-a}\right)^2\geq  \left(\frac{a-k}{a}\right)^2
$$

and follows that,
$$
\sum_{k = 1}^n y_k^N \binom{n}{k}\left(k-2 \right)   \prod_{j=1}^t \left(\frac{r_j-k}{r_j}\right)^2 \geq
\sum_{k = 1}^n   y_k^N \binom{n}{k} \left(k-2 \right)  \prod_{j=1}^t \left(\frac{r'_j-k}{r'_j}\right)^2\\
$$

Now, let $P(k)$ be the univariate polynomial with $r_1 \in (0,1)$ and $r_j\geq 1$, for  $j>1$. Let $P'(k)$ be the univariate polynomial with $r_1=1$ and $r'_j = r_j, j>1$.

Since for any $k\in N$
$$
 \left(\frac{r_1-k}{r_1}\right)^2\geq  \left(\frac{1-k}{1}\right)^2
$$

and follows that,
$$
\sum_{k = 1}^n  y_k^N \binom{n}{k}\left(k-2 \right)   \prod_{j=1}^t \left(\frac{r_j-k}{r_j}\right)^2 \geq
\sum_{k = 1}^n   y_k^N \binom{n}{k} \left(k-2 \right)  \prod_{j=1}^t \left(\frac{r'_j-k}{r'_j}\right)^2\\
$$

Next, let $P(k)$ be the univariate polynomial with $r_t =an$ for $a>1$ and $r_j \in [1,n]$, for  $j\neq t$. Let $P'(k)$ be the univariate polynomial with $r_t=n$ and $r'_j = r_j, j\neq t$.

Since for any $k\in N$
$$
 \left(\frac{an-k}{an}\right)^2\geq  \left(\frac{n-k}{n}\right)^2
$$

and follows that,
$$
\sum_{k = 1}^n  y_k^N \binom{n}{k}\left(k-2 \right)   \prod_{j=1}^t \left(\frac{r_j-k}{r_j}\right)^2 \geq
\sum_{k = 1}^n   y_k^N \binom{n}{k} \left(k-2 \right)  \prod_{j=1}^t \left(\frac{r'_j-k}{r'_j}\right)^2\\
$$

\item[(c)]  Let $P(k)$ be the univariate polynomial with degree $s<t$ with all real roots. Let $P'(k)$ be the polynomial of degree $t$ with all real roots such that $r'_j = r_j, j\leq s$ and $r'_j=n$ for $s<j\leq t$

For any $k\in N$, we have

$$
1 \geq \left(\frac{n-k}{n}\right)^2
$$
Hence,
$$
\left(\frac{r_1-k}{r_1}\right)^2\cdots \left(\frac{r_{s}-k}{r_{s}}\right)^2 \geq \left(\frac{r_1-k}{r_1}\right)^2\cdots \left(\frac{r_{s}-k}{r_{s}}\right)^2 \left(\frac{n-k}{n}\right)^{2(t-s)}
$$
and finally
$$
\sum_{k = 1}^n  y_k^N \binom{n}{k}\left(k-2 \right)   \prod_{j=1}^t \left(\frac{r_j-k}{r_j}\right)^2 \geq
\sum_{k = 1}^n   y_k^N \binom{n}{k} \left(k-2 \right)  \prod_{j=1}^t \left(\frac{r'_j-k}{r'_j}\right)^2\\
$$

\end{enumerate}
\end{proof}

\subsection{Further Results}\label{sect:further}
%
In a recent paper~\cite{KurpiszLM15b} the authors characterize the class of the initial 0/1 relaxations that are \emph{maximally hard} for the SoS hierarchy. Here, maximally hard means those relaxations that still have an integrality gap even after $n-1$ rounds of the SoS hierarchy.\footnote{Recall that at level $n$ the integrality gap vanishes.} An illustrative natural member of this class is given by the simple LP relaxation for the \textsc{Min Knapsack} problem, i.e.
$$(LP) ~ \min \left\{ \sum_{j =1}^{n} x_j : \\ \sum_{j =1}^{n} x_j \geq P, x \in [0,1]^n \right\}$$
In~\cite{KurpiszLM15b} it is shown that at level $n-1$ the integrality gap is $k$, for any $k\geq 2$ if and only if $P= \Theta(k) \cdot 2^{2n}$. A natural question is to understand if the SoS hierarchy is able to reduce the gap when $P$ is ``small''. {\col
{\col

This problem, for $P=1/2$, was considered by Cook and Dash \cite{cook2001matrix} as an
example where the Lovasz-Schrijver hierarchy rank is $n$. Laurent~\cite{Laurent03} showed that the Sherali-Adams hierarchy rank is
also equal to $n$ and raised the open question to find the rank for the Lasserre hierarchy.
She also showed that when $n = 2$, the Lasserre relaxation has an integrality gap at level 1, but leaves open whether or not this happens
at level $n - 1$ for a general $n$.
 In~\cite{KurpiszLM15b} it is ruled out the possibility that the Lasserre/SoS rank is $n$ for $n\geq 3$.

The following lemma provides a feasible solution for $\SoS_t(LP)$ with integrality gap arbitrarily close to $P$ when $t=\Omega(\log^{1-\eps} n)$ and for any $P<1$. The proof is omitted since it is similar to the proof of Theorem \ref{thm:knapsack_gap_main}.
}
\begin{theorem}
\label{thm:knapsack_gap_P}
For any $\delta >0$ and sufficiently large $n'$, let $t = \lfloor \log^{1-\delta} n'\rfloor$, $n=\lfloor \frac{n'}{t} \rfloor t$ and $\epsilon =o(t^{-1})$. Then the following solution is feasible for $\SoS_t(LP^+)$ with integrality gap arbitrarily close to $P$.
\begin{eqnarray}\label{eq:knapsack_gap_solution}
\y_I= \binom{n}{|I|}^{-1} \cdot
\left\{ \begin{array}{ll}
\frac{(1 + \epsilon)}{P\lfloor \log n\rfloor} & \text{for }|I|=\lfloor \log n\rfloor\\
\frac{\epsilon t}{jn}  & \text{for }|I|= j\frac{n}{t} \text{ and } j\in[t]\\
1-\sum_{\emptyset \neq I \subseteq N} \y_I & \text{for }I=\emptyset\\
0 & \text{otherwise}
\end{array}
\right.
\end{eqnarray}
\end{theorem}
}
{\col


\section{Proof of Theorem~\ref{thm:PSD_as_polynomial}}\label{sect:th1}

Theorem~\ref{thm:PSD_as_polynomial} is actually a corollary of a stronger statement (see Theorem~\ref{thm:PSD_as_polynomial_full} below) that provides \emph{necessary and sufficient} conditions for the matrix \eqref{matrixM} being positive-semidefinite. 

Theorem~\ref{thm:PSD_as_polynomial_full} uses a special family of polynomials $G_{h}(k)\in \R[k]$ whose definition is deferred to a later section (see Definition \ref{def:G(k)} in Section~\ref{sect:proof_full}). We postpone the definition because it will become natural in the flow of the proof of Theorem~\ref{thm:PSD_as_polynomial_full}.
%
Here we remark that the polynomials $G_h(k)$ of Definition \ref{def:G(k)} satisfy the conditions \eqref{eq:degree}, \eqref{eq:zeros} and \eqref{eq:geq0} of Theorem~\ref{thm:PSD_as_polynomial} (as shown in Lemma~\ref{lemma:G_h} to follow). 
%

\begin{theorem} \label{thm:PSD_as_polynomial_full}
Let $z^N_k \in \mathbb{R}$ for $k \in \{0,
\ldots,n\}$. Then for any $t\in N$ the following matrix is positive-semidefinite
\begin{equation}\label{matrixM_full}
 \sum_{k = 0}^n z^N_k \sum_{\substack{I \subseteq N \\ |I| = k}} Z_IZ_I^\top  \qquad (\text{where } Z_I \in \mathbb{R}^{\PS_{t}(N)})
\end{equation}
if and only if
\begin{equation}\label{eq:sym_psd_cond_full}
\sum_{k=0}^n z^N_k \binom{n}{k} G_h(k) \geq 0  \qquad \text{ for } h \in \{0,\ldots,t\}
\end{equation}
for every univariate polynomial $G_h(x)\in\mathbb{R}[x]$ of degree at most $2t$ as defined in Definition~\ref{def:G(k)}.
\end{theorem}
By Lemma \ref{lemma:G_h}, Theorem \ref{thm:PSD_as_polynomial} is a straightforward  corollary of Theorem~\ref{thm:PSD_as_polynomial_full}. In the following we provide a proof for the latter.
}

\subsection{Proof of Theorem~\ref{thm:PSD_as_polynomial_full}}\label{sect:proof_full}

We
study when the matrix
$
M= \sum_{k=0}^n  z_k \sum_{I\subseteq N, |I|=k}Z_I Z_I^{\top},
$
where $Z_I \in \mathbb{R}^{\PS_t(N)}$
is positive-semidefinite. Theorem~\ref{thm:PSD_as_polynomial_full} allows us to reduce the condition $M \succeq 0$ to inequalities of the form
$
\sum_{k=0}^n \binom{n}{k}z_kp(k) \geq 0
$,
where $p(k)$ is a univariate polynomial of degree $2t$ with some additional remarkable properties. 

A basic key idea that is used to obtain such a characterization is that the eigenvectors of $M$ are ``very well'' structured. This structure is used to get $p(k)$ with the claimed  properties.
\paragraph{The structure of the eigenvectors.}
Let $\Pi$ denote the group of all permutations of the set $N$, i.e. the \emph{symmetric} group. Let $P_\pi$ be the permutation matrix of size $ \PS_t(N)\times \PS_t(N)$ corresponding to any permutation $\pi$ of set $N$, i.e. for any vector $v$ we have $[P_\pi v]_I = v_{\pi(I)}$ for any $I\in \PS_t(N)$ (see Footnote \ref{footnote:set_permutation}). Note that $P_\pi^{-1} = P_\pi^\top$.
%

\begin{lemma} \label{lemma:invariant_aut}
For every $\pi\in \Pi$ we have
$
P^\top_\pi M P_\pi = M
$
or, equivalently, $M$ and $P_\pi$ commute $M P_\pi = P_\pi M$.
\end{lemma}
\begin{proof}
Let $e_I$ denote the vector with a $1$ in the $I$-th coordinate and $0$'s elsewhere.
Observe
$$
P_\pi^\top Z_I = P_\pi^\top \sum_{Q \subseteq I}e_Q = \sum_{Q \subseteq I}P^\top_\pi e_Q = \sum_{Q \subseteq I}e_{\pi^{-1}(Q)} = \sum_{\pi(H) \subseteq I}e_H = \sum_{H \subseteq \pi^{-1}(I)}e_H = Z_{\pi^{-1}(I)}
$$
Then,
$
P_\pi^\top M P_\pi = \sum_{k=0}^n z_k\sum_{I\subseteq N}P_\pi^\top Z_I Z_I^{\top} P_\pi = \sum_{k=0}^n z_k \sum_{I\subseteq N} Z_{\pi^{-1}(I)}Z_{\pi^{-1}(I)}^\top = M
$.

\end{proof}
\begin{corollary} \label{lemma:eigenvalues}
 If $w \in \mathbb{R}^{\PS_t(N)}$ is an eigenvector of $M$  then $v=P_\pi w$ is also an  eigenvector of $M$ for any $\pi \in \Pi$.
\end{corollary}
\begin{proof}
By the assumption, $Mw=\lambda w$ and by Lemma~\ref{lemma:invariant_aut},
$
Mv = M (P_\pi w) = P_\pi M w = \lambda v
$.

\end{proof}


%

By using Corollary~\ref{lemma:eigenvalues} 
 we can show that the set of interesting eigenvectors have some ``strong'' symmetry properties that will be used in our analysis. In the simplest case, for any eigenvector $w$ we could take the vector $u = \sum_{\pi\in \Pi} P_\pi w $ and observe that the elements of $u$ have the form $u_I = u_J$ for each $I, J$ such that $|I| = |J|$. If $\|u\|\not= 0$ then $u/\|u\|$ and $w/\|w\|$ are two eigenvectors corresponding to the same eigenvalue. The latter implies that by considering \emph{only} eigenvectors having the form $u_I = u_J$ for each $|I|=|J|$ we would consider the eigenvalue corresponding to the ``unstructured'' eigenvector $w$ as well. This is not the case in general, however, since it is possible that $\sum_{\pi\in \Pi} P_\pi w = 0$.

We overcome this obstacle by restricting the permutations in a way which guarantees $u$ to be non-zero. Before going into the details, we introduce some notation.
\begin{definition}
For any $H \subseteq N$, we denote by $\Pi_H$ the permutation group that \emph{fixes} the set $H$ in the following sense:
$
\pi \in \Pi_H \Leftrightarrow \pi(H) = H
$.
\end{definition}
Note that the definition is equivalent to saying that $\pi \in \Pi_H$ if and only if $\pi(i) \in H$ for every $i \in H$ and $\pi(i) \notin H$ for every $i \notin H$.

Now, we choose a subset $H \subseteq N$ such that $\sum_{\pi \in \Pi_I} P_\pi w = 0$ for each $I$ such that $|I| < |H|$ and $u = \sum_{\pi \in \Pi_H} P_\pi w \neq 0$. Such a set $H$ always exists, since otherwise $w$ is a zero vector, since if there is one non-zero entry $w_J$ in $w$, we can take $H = J$ and the resulting $u$ is non-zero.
The choice of $H$ is not unique, but we can always assume that it is the subset of the first $h=|H|$ elements from $N$, i.e. $H=\{1,\ldots,h\}$. Indeed, if it is not the case, there exists a permutation $\pi\in \Pi$ that maps $H$ to the subset of the first $|H|$ elements from $N$ and such that $P_{\pi} w$ is an eigenvector of $M$ by Lemma~\ref{lemma:eigenvalues}.
Now it holds that $u \neq 0$ and the vector $u/\|u\|$ is a unit eigenvector corresponding to the same eigenvalue as $w$ and has many elements that are equal to each other.


\begin{lemma} \label{lemma:permuted_eigenvector}
Let $w \in \mathbb{R}^{\PS_t(N)}$ be a unit eigenvector of $M$ corresponding to eigenvalue $\lambda$, and $H$ be the smallest subset of $N$ such that
$
u = \sum_{\pi \in \Pi_H} P_\pi w \neq 0
$.
Then $u/\|u\|$ is also a unit eigenvector of $M$ corresponding to eigenvalue $\lambda$.
\end{lemma}

The following lemma shows the structure of eigenvectors obtained from summing the permutations of any ``unstructured'' eigenvector.

\begin{lemma} \label{lemma:symmetry_of_u}
Let $
u = \sum_{\pi \in \Pi_H} P_\pi w
$. Then the vector $u$ is invariant under the permutations of $\Pi_H$, namely $u_I=u_{\pi(I)}$ for $\pi\in \Pi_H$. Equivalently, $u_I = u_J$ for all $|I|=|J|$ such that $|I \cap H| = |J \cap H|$.
\end{lemma}
\begin{proof}
 We have
$
u_I = \left[\sum_{\pi \in \Pi_H} P_\pi w \right]_I = \sum_{\pi \in \Pi_H} [P_\pi w]_I = \sum_{\pi \in \Pi_H} w_{\pi(I)}= \sum_{\pi \in \Pi_H} w_{\pi(\pi(I))}=u_{\pi(I)}
$ where the last but one equality follows since permutations are bijective.
The claim follows by observing that for all $|I|=|J|$ such that $|I \cap H| = |J \cap H|$ there exists $\pi\in \Pi_H$ such that $\pi(I)=J$.

\end{proof}
Lemma~\ref{lemma:permuted_eigenvector}, Lemma~\ref{lemma:symmetry_of_u} and the arguments above imply Lemma~\ref{lemma:eigenvectors}.
\begin{lemma}\label{lemma:eigenvectors}
For any eigenvalue $\lambda$ of $M$ there exists an $h=0,1,\dots,t$ such that the following is an eigenvector corresponding to $\lambda$:
\begin{equation}\label{eq:eigenform}
u_h = \sum_{i=0}^t \sum_{j=0}^{\min\{h,i\}} \alpha_{i,j} b_{i,j}
\end{equation}
where $H=\{1,\ldots,h\}$, $\alpha_{i,j} \in \R$ and $b_{i,j} \in \mathbb{R}^{\PS_t(N)}$ such that $[b_{i,j}]_Q = 1$ if $|Q| = i$ and $|Q \cap H| = j$, $[b_{i,j}]_Q = 0$ otherwise.
\end{lemma}
By Lemma~\ref{lemma:eigenvectors}, we have that the positive semidefiniteness condition of $M$ follows by ensuring that for any $h=0,1,\ldots,t$ we have $u_h^\top M u_h \geq 0$, i.e.
$$ u_h^\top M u_h=\sum_{k=0}^n  z_k \underbrace{\sum_{\substack{I\subseteq N\\ |I|=k}}\left(u_h^\top Z_I\right)^2}_{A_k} =   \sum_{k=0}^n z_k \sum_{\substack{I\subseteq N\\ |I|=k}} \left(\sum_{i=0}^t \sum_{j=0}^{\min\{h,i\}} \alpha_{i,j} b_{i,j}^\top Z_I\right)^2 \geq 0
$$

In the following (Lemma~\ref{lemma:polyform}) we show that the above values $A_k$ are interpolated by the univariate polynomial $G_{h}(x)$ defined in Definition~\ref{def:G(k)}. In Lemma~\ref{lemma:G_h} we prove some remarkable properties of $G_{h}(x)$ as claimed in Theorem~\ref{thm:PSD_as_polynomial}.

\begin{definition}\label{def:G(k)}
For any $h\in\{0,\ldots,t\}$, let $G_{h}(k)\in \R[k]$ be a univariate polynomial defined as follows
\begin{align}
G_h(k)=\sum_{r = 0}^{h} \binom{h}{r}h_r(k) \left(\sum_{j=0}^{h}\binom{r}{j} p_j(k-r) \right)^2
\end{align}
where $h_r(k) =  \fallfac{k}{r} \cdot \fallfac{(n-k)}{h-r}$ and $p_j(k-r) = \sum_{i=0}^{t-j} \alpha_{i+j,j}\binom{k-r}{i}$ {\col (for $\alpha_{i,j} \in \R$)}.\footnote{Denote by $\fallfac{x}{m} = x(x-1)\cdots(x-m+1)$ the falling factorial (with the convention that $\fallfac{x}{0} = 1$).}
\end{definition}

%
\begin{lemma}\label{lemma:polyform} For every $k=0,\ldots,n$ the following identity holds
$
A_k=\binom{n}{k} \frac{1}{\fallfac{n}{h} } G_{h}(k)
$.
\end{lemma}

\begin{proof}

We start noting that for every $i = 0,...,t$, $j = 0,...,|H|$ we have~\footnote{Recall that $\binom{n}{-k}=\binom{n}{n+k}=0$ for any positive integer $k$.}
$$
b_{i,j}^\top Z_I = \binom{|I \cap H|}{j}\binom{|I \setminus H|}{i-j}
$$
Indeed
 \begin{eqnarray*}
 b_{i,j}^\top Z_I = \sum_{Q \in \PS_t(N)} (b_{i,j})_Q (z_I)_Q
= \sum_{\substack{Q \subseteq I, |Q|=i \\ |Q \cap H| = j}} 1 = \binom{|I \cap H|}{j}\binom{|I \setminus H|}{i-j}
 \end{eqnarray*}
It follows that we have
\begin{align*}
\sum_{\substack{I\subseteq N\\ |I|=k}} \left(u^\top Z_I\right)^2 =  \sum_{\substack{I\subseteq N\\ |I|=k}} \left(\sum_{i=0}^t \sum_{j=0}^{|H|} \alpha_{i,j} b_{i,j}^\top Z_I\right)^2 \\
= \sum_{\substack{I\subseteq N\\ |I|=k}} \left(\sum_{i=0}^t \sum_{j=0}^{|H|} \alpha_{i,j} \binom{|I \cap H|}{j}\binom{|I \setminus H|}{i-j} \right)^2
\end{align*}
Splitting the sum over $I$ by considering the intersections $I \cap H$ of sizes $r = 0,...,|H|$, we have
\begin{align*}
\sum_{r = 0}^{|H|}\sum_{\substack{|I|=k \\ |I\cap H| = r}} \left(\sum_{i=0}^t \sum_{j=0}^{|H|} \alpha_{i,j} \binom{r}{j}\binom{k-r}{i-j} \right)^2\\
 = \sum_{r = 0}^{|H|} \binom{|H|}{r} \binom{n-|H|}{k-r} \left(\sum_{j=0}^{|H|}\binom{r}{j} \sum_{i=0}^t \alpha_{i,j}\binom{k-r}{i-j} \right)^2
\end{align*}
Finally, we shift the sum over $i$ by $j$ and thus justify the equality
$$
\sum_{\substack{I\subseteq N\\ |I|=k}} \left(u^\top Z_I\right)^2 = \sum_{r = 0}^{|H|} \binom{|H|}{r} \binom{n-|H|}{k-r} \left(\sum_{j=0}^{|H|}\binom{r}{j} \sum_{i=0}^{t-j} \alpha_{i+j,j}\binom{k-r}{i} \right)^2
$$
Now, the sum over $i$ is a Newton polynomial that we denote by $p_j(k-r) = \sum_{i=0}^{t-j} \alpha_{i+j,j}\binom{k-r}{i}$. Note that by definition, here $\deg(p) = t-j$. Furthermore, observe that
$$
\binom{n-|H|}{k-r} = \binom{n}{k} \frac{1}{\fallfac{n}{|H|} } \fallfac{k}{r} \cdot \fallfac{(n-k)}{|H|-r}
$$
and writing $h_r(k) = \fallfac{k}{r} \cdot \fallfac{(n-k)}{|H|-r}$ yields the claim.
\end{proof}

{\col
It follows that for any unit eigenvector $u$ of the form~\eqref{eq:eigenform} the corresponding eigenvalue is equal~to
$
u^\top M u = \frac{1}{\fallfac{n}{h}} \sum_{k=0}^n z_k \binom{n}{k} G_{h}(k)
$. 
Theorem \ref{thm:PSD_as_polynomial_full} requires that $\sum_{k=0}^n z_k \binom{n}{k} G_{h}(k)\geq 0$ which implies that the eigenvalue $u^\top M u$ is nonnegative. 
}
In the following section we complete the proof by showing that the polynomials $G_h(k)$ of Definition \ref{def:G(k)} satisfy the conditions \eqref{eq:degree}, \eqref{eq:zeros} and \eqref{eq:geq0} of Theorem~\ref{thm:PSD_as_polynomial} (as shown in Lemma~\ref{lemma:G_h}).
\subsection{Properties of the univariate polynomials}
\begin{lemma}\label{lemma:G_h}
 For any $h\in\{0,\ldots,t\}$, the polynomials $G_h(k)$ as defined in Definition \ref{def:G(k)} have the following properties:
 \begin{itemize}
  \item[(a)] $G_h(k)$ is a univariate polynomial of degree at most $2t$,
  \item[(b)] $G_h(k) \geq 0$ for $k \in [h-1,n-h+1]$
  \item[(c)] $G_h(k) = 0$ for every $k \in \set{0,...,h-1} \cup \set{n-h+1,...,n}$.
 \end{itemize}
\end{lemma}
\paragraph{Proof of \emph{(a)}.}
\begin{eqnarray*}
G_{h}(k) &=& \sum_{r = 0}^{h} \binom{h}{r}h_r(k) \left(\sum_{j=0}^{h}\binom{r}{j} p_j(k-r) \right)^2\\
&=&\sum_{r = 0}^{h} \binom{h}{r}h_r(k) \left(\sum_{i=0}^{h}\sum_{j=0}^{h}\binom{r}{i}\binom{r}{j} p_i(k-r)p_j(k-r)\right)\\
&=&\sum_{r = 0}^{h} \binom{h}{r}h_r(k) \left(\sum_{i=0}^{h}\sum_{j=0}^{h}\binom{r}{i}\binom{r}{j} \left(\sum_{a=0}^{t-i} \alpha_{a+i,i}\binom{k-r}{a}\right)\left(\sum_{b=0}^{t-j} \alpha_{b+j,j}\binom{k-r}{b}\right)\right)\\
&=&\sum_{i=0}^{h}\sum_{j=0}^{h} \sum_{a=0}^{t-i}\sum_{b=0}^{t-j}  \alpha_{a+i,i}\alpha_{b+j,j} \sum_{q=0}^a \sum_{s=0}^b \underbrace{\binom{k-h}{q}\binom{k-h}{s}\left(\sum_{r = 0}^{h} \binom{h}{r}h_r(k) \binom{r}{i}\binom{r}{j} \binom{h-r}{a-q}\binom{h-r}{b-s}\right)}_{B(k)}
\end{eqnarray*}

Note that $\binom{k-r}{a}=\sum_{q=0}^a \binom{k-h}{q}\binom{h-r}{a-q}$ by Vandermonde's identity.
We prove the theorem by showing that $B(k)$ has degree not larger than $2t$.
\begin{eqnarray*}
B(k)&=&\binom{k-h}{q}\binom{k-h}{s}\underbrace{\left(\sum_{r = 0}^{h} \binom{h}{r}\fallfac{k}{r} \fallfac{(n-k)}{h-r} \overbrace{\binom{r}{i}\binom{r}{j} \binom{h-r}{a-q}\binom{h-r}{b-s}}^{f(r)}\right)}_{C(k)}
\end{eqnarray*}
By Lemma~\ref{th:degree} below, the degree of $C(k)$ is at most $i+j+a-q+b-s$ and thus the degree of $B(k)$ is at most $i+j+a+b=2t$

%
\begin{lemma}\label{th:degree}
The degree of $C(k)$ is at most $i+j+a-q+b-s$.
\end{lemma}

The claim follows by showing that the degree of $C(k)$ is at most the degree of $f( r )$.
The degree of $f( r )$ is $i+j+a-q+b-s$.

Recall that the forward difference of function $g(X)$ with respect to variable $X$ is a finite difference defined by $\Delta_X[g(X)] = g(X+1)-g(X)$. Higher order differences are obtained by repeated operations of the forward difference operator. We will use $\Delta_X^{\ell}[g(X)]_{X=b}$ to denote the $\ell$-th forward difference evaluated at $X=b$. {\col We will us the following easy to check identity:  $\Delta_X^{d}[\fallfac{(k+X)}{r+d}]=\fallfac{(k+X)}{r}\fallfac{(r+d)}{d}$.}

First note that any polynomial $f( r )$ of degree $\delta$ can written as linear combinations of polynomials $P_d( r )=\risefac{(r+1)}{d}=\fallfac{(r+d)}{d}$ with $0\leq d \leq \delta$. It follows that the claim follows by showing that the degree of the following $C'(k)$ is at most the degree of $P_d( r )$
\begin{eqnarray*}
C'(k) &=& \sum_{r = 0}^{h} \binom{h}{r} \fallfac{(n-k)}{h-r} \cdot \fallfac{k}{r} \fallfac{(r+d)}{d}\\
&=& \sum_{r = 0}^{h} \binom{h}{r} \fallfac{(n-k)}{h-r} \cdot \Delta_X^d\left[ \fallfac{(k+X)}{r+d}\right]_{X=0}\\
&=&  \Delta_X^d\left[  \sum_{r = 0}^{h} \binom{h}{r} \fallfac{(n-k)}{h-r}\fallfac{(k+X)}{r+d}\right]_{X=0}\\
&=&  \Delta_X^d\left[  \fallfac{(k+X)}{d} \sum_{r = 0}^{h} \binom{h}{r} \fallfac{(n-k)}{h-r}\fallfac{(k+X-d)}{r}\right]_{X=0}\\
&=&  \Delta_X^d\left[  \fallfac{(k+X)}{d} \fallfac{(n+X-d)}{h}\right]_{X=0}\\
\end{eqnarray*}
where we have used the linearity of the forward difference operator and the Vandermonde's identity to derive the last equality.
The claim follows by observing that the forward difference operator does not increase the degree of its argument and therefore $C'(k)$ has degree at most~$d$.

\paragraph{Proof of \emph{(b)}.}
Let $k \in [h-1,n-h+1]$. We have
$$
G_h(k)=\sum_{r = 0}^{h} \binom{h}{r}h_r(k) \left(\sum_{j=0}^{h}\binom{r}{j} p_j(k-r) \right)^2
$$
where $h_r(k) = \fallfac{k}{r} \cdot \fallfac{(n-k)}{h-r} \geq 0$ for each $r = 0,...,h$. Therefore $G_h(k)$ is a sum of non-negative numbers $\left(\sum_{j=0}^{h}\binom{r}{j} p_j(k-r) \right)^2$ weighted by positive coefficients $h_r(k)$.

\paragraph{Proof of \emph{(c)}.}

From Lemma~\ref{lemma:polyform} we have that
$$
\frac{1}{\fallfac{n}{h}}\binom{n}{k} G_h(k) = \sum_{\substack{I\subseteq N\\ |I|=k}} \left(u^\top Z_I\right)^2
$$
Therefore $G_{h}(k) = 0$ for $k \in \set{0,...,h-1} \cup \set{n-h+1,...,n}$ if we can show that $u^\top Z_Q = 0$ for all $Q \subseteq N$ such that $|Q| = k$.

With this aim, we start noting that for every set $S\subseteq N$ we have that the permutation group $ \Pi_{S}$ is the same as $\Pi_{N\setminus S}$. Moreover if $|S| < h$ then  $\sum_{\pi \in \Pi_{S}} P_\pi u= 0$, otherwise we would obtain a set $S$ smaller than $H$ with $\sum_{\pi \in \Pi_{S}} P_\pi u\not = 0$ (contradicting our assumption that $H$ is a set with the smallest set size having $\sum_{\pi \in \Pi_{H}} P_\pi u\not = 0$).

Now consider any set $I$ such that $I\subseteq Q$ with $Q\in \{S,N\setminus S\}$ and $|S| < h$. By the previous observations it follows that $[\sum_{\pi \in \Pi_{Q}} P_\pi u]_I=\sum_{\pi \in \Pi_{Q}} P_\pi u_I =\sum_{\pi \in \Pi_{Q}} u_{\pi(I)}= 0$. Note that since $I\subseteq Q$ the set $\{\pi(I): \pi \in \Pi_{Q}\}$ is equal to $\{J: J\subseteq Q, |J|=|I|\}$, since $\Pi_{Q}$ is the permutation group that maps any element $I$ from $Q$ to any other element from $Q$ of the same size.
It follows that $\sum_{\pi \in \Pi_{Q}} u_{\pi(I)} = \sum_{J\subseteq Q, |J|=|I|} u_{J} = 0$. Using the latter we get
$$
u^\top Z_Q  = \sum_{J \subseteq Q} u_J = \sum_{i=0}^{|Q|}\sum_{J \subseteq Q, |J| = i}u_J = 0
$$
proving the claim.

\paragraph{Acknowledgements.}
The authors would like to express their gratitude to Ola Svensson for helpful discussions and ideas regarding this paper.

{\small
\bibliographystyle{abbrv}
\bibliography{min_knapsack_gap}

\begin{thebibliography}{10}

\bibitem{BansalBN08}
N.~Bansal, N.~Buchbinder, and J.~Naor.
\newblock Randomized competitive algorithms for generalized caching.
\newblock In {\em STOC}, pages 235--244, 2008.

\bibitem{BansalGK10}
N.~Bansal, A.~Gupta, and R.~Krishnaswamy.
\newblock A constant factor approximation algorithm for generalized min-sum set
  cover.
\newblock In {\em SODA}, pages 1539--1545, 2010.

\bibitem{BansalP10}
N.~Bansal and K.~Pruhs.
\newblock The geometry of scheduling.
\newblock In {\em FOCS}, pages 407--414, 2010.

\bibitem{BarakBHKSZ12}
B.~Barak, F.~G. S.~L. Brand{\~a}o, A.~W. Harrow, J.~A. Kelner, D.~Steurer, and
  Y.~Zhou.
\newblock Hypercontractivity, sum-of-squares proofs, and their applications.
\newblock In {\em STOC}, pages 307--326, 2012.

\bibitem{BarakCK15}
B.~Barak, S.~O. Chan, and P.~Kothari.
\newblock Sum of squares lower bounds from pairwise independence.
\newblock In {\em STOC}, 2015.

\bibitem{BarakS14}
B.~Barak and D.~Steurer.
\newblock Sum-of-squares proofs and the quest toward optimal algorithms.
\newblock {\em Electronic Colloquium on Computational Complexity {(ECCC)}},
  21:59, 2014.

\bibitem{BhaskaraCVGZ12}
A.~Bhaskara, M.~Charikar, A.~Vijayaraghavan, V.~Guruswami, and Y.~Zhou.
\newblock Polynomial integrality gaps for strong sdp relaxations of densest
  {\it k}-subgraph.
\newblock In {\em SODA}, pages 388--405, 2012.

\bibitem{BlekhermanGP14}
G.~Blekherman, J.~ao~Gouveia, and J.~Pfeiffer.
\newblock Sums of squares on the hypercube.
\newblock {\em CoRR}, abs/1402.4199, 2014.

\bibitem{CarnesS08}
T.~Carnes and D.~B. Shmoys.
\newblock Primal-dual schema for capacitated covering problems.
\newblock In {\em IPCO}, pages 288--302, 2008.

\bibitem{CarrFLP00}
R.~D. Carr, L.~Fleischer, V.~J. Leung, and C.~A. Phillips.
\newblock Strengthening integrality gaps for capacitated network design and
  covering problems.
\newblock In {\em SODA}, pages 106--115, 2000.

\bibitem{ChakrabartyGK10}
D.~Chakrabarty, E.~Grant, and J.~K{\"o}nemann.
\newblock On column-restricted and priority covering integer programs.
\newblock In {\em IPCO}, pages 355--368, 2010.

\bibitem{Cheung07}
K.~K.~H. Cheung.
\newblock Computation of the {L}asserre ranks of some polytopes.
\newblock {\em Mathematics of Operation Research}, 32(1):88--94, 2007.

\bibitem{Chla12}
E.~Chlamtac and M.~Tulsiani.
\newblock Convex relaxations and integrality gaps.
\newblock In {\em Handbook on Semidefinite, Conic and Polynomial Optimization},
  volume 166, pages 139--169. Springer, 2011.

\bibitem{cook2001matrix}
W.~Cook and S.~Dash.
\newblock On the matrix-cut rank of polyhedra.
\newblock {\em Mathematics of Operations Research}, 26(1):19--30, 2001.

\bibitem{FawziSaundersonParrilo15}
H.~Fawzi, J.~Saunderson, and P.~Parrilo.
\newblock Sparse sum-of-squares certificates on finite abelian groups.
\newblock {\em CoRR}, abs/1503.01207, 2015.

\bibitem{Godsil10}
C.~Godsil.
\newblock Association schemes.
\newblock Lecture Notes available at
  http://quoll.uwaterloo.ca/mine/Notes/assoc2.pdf, 2010.

\bibitem{GoemansT01}
M.~X. Goemans and L.~Tun{\c{c}}el.
\newblock When does the positive semidefiniteness constraint help in lifting
  procedures?
\newblock {\em Math. Oper. Res.}, 26(4):796--815, 2001.

\bibitem{Grigoriev01}
D.~Grigoriev.
\newblock Complexity of positivstellensatz proofs for the knapsack.
\newblock {\em Computational Complexity}, 10(2):139--154, 2001.

\bibitem{Grigoriev01b}
D.~Grigoriev.
\newblock Linear lower bound on degrees of positivstellensatz calculus proofs
  for the parity.
\newblock {\em Theoretical Computer Science}, 259(1-2):613--622, 2001.

\bibitem{GrigorievHP02}
D.~Grigoriev, E.~A. Hirsch, and D.~V. Pasechnik.
\newblock Complexity of semi-algebraic proofs.
\newblock In {\em STACS}, pages 419--430, 2002.

\bibitem{HongT08}
S.~Hong and L.~Tun{\c{c}}el.
\newblock Unification of lower-bound analyses of the lift-and-project rank of
  combinatorial optimization polyhedra.
\newblock {\em Discrete Applied Mathematics}, 156(1):25--41, 2008.

\bibitem{Khot02a}
S.~Khot.
\newblock On the power of unique 2-prover 1-round games.
\newblock In {\em STOC}, pages 767--775, 2002.

\bibitem{KurpiszLM15b}
A.~Kurpisz, S.~Lepp{\"{a}}nen, and M.~Mastrolilli.
\newblock A lasserre lower bound for the min-sum single machine scheduling
  problem.
\newblock In {\em ESA}, pages 853--864, 2015.

\bibitem{KurpiszLM15}
A.~Kurpisz, S.~Lepp{\"{a}}nen, and M.~Mastrolilli.
\newblock On the hardest problem formulations for the 0/1 lasserre hierarchy.
\newblock In {\em ICALP}, pages 872--885, 2015.

\bibitem{Lasserre01}
J.~B. Lasserre.
\newblock Global optimization with polynomials and the problem of moments.
\newblock {\em SIAM Journal on Optimization}, 11(3):796--817, 2001.

\bibitem{Laurent03}
M.~Laurent.
\newblock A comparison of the {S}herali-{A}dams, {L}ov{\'a}sz-{S}chrijver, and
  {L}asserre relaxations for 0-1 programming.
\newblock {\em Mathematics of Operations Research}, 28(3):470--496, 2003.

\bibitem{Laurent03a}
M.~Laurent.
\newblock Lower bound for the number of iterations in semidefinite hierarchies
  for the cut polytope.
\newblock {\em Math. Oper. Res.}, 28(4):871--883, 2003.

\bibitem{LeeRagSteu15}
J.~R. Lee, P.~Raghavendra, and D.~Steurer.
\newblock Lower bounds on the size of semidefinite programming relaxations.
\newblock In {\em STOC}, pages 567--576, 2015.

\bibitem{MekaPW15}
R.~Meka, A.~Potechin, and A.~Wigderson.
\newblock Sum-of-squares lower bounds for planted clique.
\newblock In {\em STOC}, pages 87--96, 2015.

\bibitem{ODonnell13}
R.~O'Donnell.
\newblock Approximability and proof complexity.
\newblock Talk at ELC Tokyo. Slides available at
  http://www.cs.cmu.edu/~odonnell/slides/approx-proof-cxty.pps, 2013.

\bibitem{parrilo00}
P.~Parrilo.
\newblock {\em Structured Semidefinite Programs and Semialgebraic Geometry
  Methods in Robustness and Optimization}.
\newblock {PhD} thesis, California Institute of Technology, 2000.

\bibitem{Raghavendra08}
P.~Raghavendra.
\newblock Optimal algorithms and inapproximability results for every csp?
\newblock In {\em STOC}, pages 245--254, 2008.

\bibitem{Rot13}
T.~Rothvo{\ss}.
\newblock The lasserre hierarchy in approximation algorithms.
\newblock Lecture Notes for the MAPSP 2013 - Tutorial, June 2013.

\bibitem{Schoenebeck08}
G.~Schoenebeck.
\newblock Linear level {L}asserre lower bounds for certain k-csps.
\newblock In {\em FOCS}, pages 593--602, 2008.

\bibitem{StephenT99}
T.~Stephen and L.~Tun{\c{c}}el.
\newblock On a representation of the matching polytope via semidefinite
  liftings.
\newblock {\em Math. Oper. Res.}, 24(1):1--7, 1999.

\bibitem{Tulsiani09}
M.~Tulsiani.
\newblock Csp gaps and reductions in the {L}asserre hierarchy.
\newblock In {\em STOC}, pages 303--312, 2009.

\bibitem{Wolsey75}
L.~A. Wolsey.
\newblock Facets for a linear inequality in 0--1 variables.
\newblock {\em Mathematical Programming}, 8:168--175, 1975.

\end{thebibliography}
}


\appendix
\section*{Appendix}

\section{The SoS hierarchy} \label{app:SoS}
In this section we recall the usual definition of the SoS/Lasserre hierarchy~\cite{Lasserre01} and justify Definition~\ref{def:sos_definition}. Notice that the SDP hierarchy that we discuss here is the dual certificate of a refutation of the positivstellensatz proof system, for further information about the connection to the proof system we refer the reader to~\cite{MekaPW15}. In our setting we restrict ourselves to problems with $0/1$-variables and linear constraints. More precisely, we consider the following general optimization problem $\mathbb{P}$: given a multilinear polynomial $f:\{0,1\}^n\rightarrow \mathbb{R}$
\begin{equation}\label{eq:polyproblem}
\mathbb{P}: \quad \min\{f(x)| x\in\{0,1\}^n \cap K\}
\end{equation}
where $K$ is a polytope defined by $m$ linear inequalities $g_{\ell}(x)\geq 0 \text{ for } \ell\in [m]$. Many basic optimization problems are special cases of $\mathbb{P}$. For example, any $k$-ary boolean constraint satisfaction problem, such as \textsc{Max-Cut}, is captured by~\eqref{eq:polyproblem} where a degree $k$ function $f(x)$ counts the number of satisfied constraints, and no linear constraints $ g_{\ell}(x)\geq 0$ are present. Also any $0/1$ integer linear program is a special case of~\eqref{eq:polyproblem}, where $f(x)$ is a linear function.

Lasserre~\cite{Lasserre01} proposed a hierarchy of SDP relaxations parameterized by an integer $r$,
\small
\begin{equation}\label{eq:lass1}
 \min\{L(f)| L: \mathbb{R}[X]_{2r}\rightarrow \mathbb{R},  L(1)=1, L(x^2-x)=0 \text{ and } L(u^2), L(u^2 g_{\ell})\geq 0,  \forall \text{ polynomial } u \}
\end{equation}
\normalsize
 where $L: \mathbb{R}[X]_{2r}\rightarrow \mathbb{R}$ is a linear map with $\mathbb{R}[X]_{2r}$ denoting the ring $\mathbb{R}[X]$ restricted to polynomials of degree at most $2r$.\footnote{In \cite{BarakBHKSZ12}, $L(p)$ is written $\tilde{ \mathbb{E}}[p]$ and called the ``pseudo-expectation'' of $p$.} Note that~\eqref{eq:lass1} is a relaxation since one can take $L$ to be the evaluation map $f\rightarrow f(x^*)$ for any optimal solution $x^*$.

Relaxation~\eqref{eq:lass1} can be equivalently formulated in terms of \emph{moment matrices}~\cite{Lasserre01}. In the context of this paper, this matrix point of view is more convenient to use and it is described below. In our notation we mainly follow the survey of Laurent~\cite{Laurent03} (see also \cite{Rot13}).

\paragraph{Variables and Moment Matrix.} Let $N$ denote the set $\{1,\ldots,n\}$. The collection of all subsets of $N$ is denoted by $\PS(N)$. For any integer $t\geq 0$, let $\PS_t(N)$ denote the collection of subsets of $N$ having cardinality at most~$t$.
Let $y\in \mathbb{R}^{\PS(N)}$. For any nonnegative integer $t\leq n$, let $M_t(y)$ denote the matrix with $(I,J)$-entry $y_{I\cup J}$ for all $I,J\in \PS_t(N)$. Matrix $M_t(y)$ is termed in the following as the \emph{t-moment matrix} of $y$. For a linear function $g(x) = \sum_{i=1}^n g_{i} \cdot x_i + g_0$, we define $g*y$ as a vector, often called \emph{shift operator}, where the $I$-th entry is $(g*y)_I=\sum_{i=1}^n g_i y_{I\cup\{i\}} + g_0 y_I$. Let $f$ denote the vector of coefficients of polynomial $f(x)$ (where $f_I$ is the coefficient of monomial $\Pi_{i\in I}x_i$ in $f(x)$).

\begin{definition}\label{def:lassDef}
The $t$-th round SoS (or Lasserre) relaxation of problem~\eqref{eq:polyproblem}, denoted as $\SoS_{t}(\mathbb{P})$, is the following
\begin{equation}
\SoS_{t}(\mathbb{P}):\quad \min\left\{\sum_{I \subseteq N} f_I y_I | y\in \mathbb{R}^{\PS_{2t+2d}(N)} \text{ and } y\in \mathbb{M} \right\}
\end{equation}
where $\mathbb{M}$ is the set of vectors $y\in \mathbb{R}^{\PS_{2t+2d}(N)}$ that satisfy the following PSD conditions
\begin{eqnarray}
y_{\varnothing}&=&1,  \\
M_{t+d}(y)&\succeq& 0,  \\
M_{t}( g_{\ell}*y )&\succeq& 0 \qquad \ell\in [m]
\end{eqnarray}
where $d=0$ if $m=0$ (no linear constraints), otherwise $d=1$.
\end{definition}

\paragraph{Change of variables.} A solution of the SoS hierarchy as defined in Definition~\ref{def:lassDef} is given by a vector $y \in \mathbb{R}^{\PS_{2t+2d}(N)}$. Next we show we can make a change of basis and replace the variables $y_I$ with other variables $\y_I$ that are indexed by \textit{all} the subsets of $N$.
Variable $\y_I$ can be interpreted as the ``relaxed'' indicator variable for the integral solution $x_I$, i.e. the $0/1$ solution obtained by setting $x_i = 1$ for $i \in I$, and $x_i = 0$ for $i\in N\setminus I$.
We use this change of basis  in order to obtain a useful decomposition of the moment matrix as a sum of rank one matrices of special kind. Here it is not necessary to distinguish between the moment matrix of the variables and constraints, hence in what follows we denote a generic vector by $w \in \mathbb{R}^{\PS_{2q}(N)}$, where $q$ is either $t$ or $t+1$.
\begin{definition}
\label{def:change_of_basis}
 Let $w \in \mathbb{R}^{\PS_{2q}(N)}$. For every $I \in \PS(N)$, define a vector $w^N \in \mathbb{R}^{\PS(N)} $ such that
$$
w_I = \sum_{I \subseteq H \subseteq N} w^N_H
$$
\end{definition}

Note that the inverse (for $|I|\leq 2t$) is
\begin{equation}
\label{def:change_of_basis2}
w_I^N=\sum_{H \subseteq N \setminus I, |H\cup I|\leq 2t} (-1)^{H} w_{I \cup H}
\end{equation}

To simplify the notation, we note that the moment matrix of the variables is structurally similar to the moment matrix of the constraints: if $z \in \mathbb{R}^{\PS_{2q}(N)}$ is a vector such that $z_I = \sum_{i = 1}^N A_{\ell i} y_{I\cup \set{i}} - b_\ell y_I$ for some $\ell$, then $[M_{t}(g_\ell * y)]_{I,J} = z_{I \cup J}$. Hence, the following lemma holds for the moment matrix of variables and constraints.

\begin{lemma} \label{lemma:M_in_zeta_vector_form}
 Let $w \in \mathbb{R}^{\PS_{2q}(N)}$, and $M \in \mathbb{R}^{\PS_{q}(N)\times \PS_{q}(N)}$ such that $M_{I,J} = w_{I \cup J}$. Then
 $$
 M = \sum_{H \subseteq N} w^N_H Z_HZ_H^\top
 $$
\end{lemma}
\begin{proof}
Since $M_{I,J} = w_{I \cup J}$, we have by the change of variables that
$$
[M]_{I,J} = \sum_{I \cup J \subseteq H \subseteq N} w^N_H = \sum_{H \subseteq N} \chi_{I \cup J}(H) w^N_H
$$
where $\chi_{I \cup J}(H)$ is the 0-1 indicator function such that $\chi_I(H) = 1$ if and only if $I \cup J \subseteq H$. On the other hand, $[Z_HZ_H^\top]_{I,J} = [Z_H]_I[Z_H]_J = 1$ if $I \cup J \subseteq H$, and 0 otherwise. Therefore $[Z_HZ_H^\top]_{I,J} = \chi_{I \cup J}(H)$.

\end{proof}
By the previous lemma it follows that given a solution by using variables $\{w_I^N\}$ we can obtain a solution with variables $\{w_I:|I|\leq 2t\}$. Viceversa, given any assignment of variables in $\{w_I:|I|\leq 2t\}$ we can find an assignment of variables in $\{w_I^N\}$ such that $M_{I,J}=w_{I\cup J}$ and $M = \sum_{H \subseteq N} w^N_H Z_HZ_H^\top$. Indeed, set $w_I^N=0$ for every $I$ such that $|I|>2t$. For the remaining ones note that for $|I|\leq 2t $ the square matrix corresponding to the following equalities
$w_I = \sum_{I \subseteq H \subseteq N} w^N_H
$ is invertible since it is upper triangular.

\begin{lemma}\cite{Laurent03} \label{lemma:change_of_variables_for_constraints}
Given $y \in \mathbb{R}^{\PS_{2t+2}(N)}$, for the vector $z_I = \sum_{i = 1}^n A_{\ell i} y_{I\cup \set{i}} - b_\ell y_I$ we have
\begin{equation} \label{eq:z^N_I}
z^N_{I} = g_\ell(x_I) y^N_I
\end{equation}
where $g_\ell(x_I) = \sum_{i = 1}^N A_{\ell i}x_i - b_\ell$ is a linear function corresponding to the constraint $\ell$, evaluated at $x_I$ such that $x_i = 1$ if $i \in I$ and $x_i = 0$ otherwise.
\end{lemma}
\begin{proof}
 We need to show that this choice of $z_I^N$ yields $z_I = \sum_{I \subseteq H \subseteq N} z^N_H$. We plug in \eqref{eq:z^N_I} to obtain
\begin{eqnarray*}
\sum_{I \subseteq H \subseteq N} z^N_H = \sum_{I \subseteq H \subseteq N} g_\ell(x_H) y^N_H = \sum_{I \subseteq H \subseteq N} \left[\sum_{i = 1}^n A_{\ell i}x_i - b_\ell \right]_{x = x_H} y^N_H \\
= \sum_{I \subseteq H \subseteq N}\left( \sum_{i = 1}^n \left[A_{\ell i}x_i\right]_{x = x_H}y^N_H - b_\ell  y^N_H \right) = \sum_{I \subseteq H \subseteq N} \sum_{i = 1}^n \left[A_{\ell i}x_i\right]_{x = x_H}y^N_H - b_\ell  y_I
\end{eqnarray*}
Here the term $\left[A_{\ell i}x_i\right]_{x = x_H}y^N_H$ is $A_{\ell i}y^N_H$ if $i \in H$ and 0 otherwise. Taking this into account and changing the order of the sums, the above becomes
$$
\sum_{i = 1}^n \sum_{I \cup \set{i} \subseteq H \subseteq N} A_{\ell i}y^N_H - b_\ell  y_I = \sum_{i = 1}^n A_{\ell i} y_{I\cup \set{i}} - b_\ell y_I
$$
which proves the claim.

\end{proof}

The above discussion together with the observation that $y_{\emptyset}=1$ implies that $\sum_{J\subseteq N} y_J^N =1$ and justifies Definition~\ref{def:sos_definition}. Finally, we remark the following
\begin{lemma}
Let $f$ denote the vector of coefficients of polynomial $f(x)$ of \eqref{eq:polyproblem}. Then the objective value of the solution $y$ is given by
$$
\sum_{I \subseteq N} f_I y_I=\sum_{I\subseteq  N} f(x_I) y_I^N
$$
\end{lemma}
\begin{proof}
 Similar lines as the proof of Lemmas \ref{lemma:M_in_zeta_vector_form} and \ref{lemma:change_of_variables_for_constraints}.

 \end{proof}

\section{Change of variables for Max-Cut} \label{Sect:app_max_cut}
Grigoriev \cite{Grigoriev01} and Laurent \cite{Laurent03a} proved that the following solution
is feasible for any $\K\leq n/2$ up to round $t\leq \lfloor \K\rfloor$ for the SoS hierarchy given in Definition~\ref{def:lassDef}:
$$y_{I} = \frac{\binom{\K}{|I|}}{\binom{n}{|I|}}  \qquad \forall I\subseteq N:|I|\leq 2t$$
Using the change of basis~\eqref{def:change_of_basis2}, solution $\{y_I\}$ is equivalent to solution $\{\y_I\}$:
\begin{align}
\y_I &= \sum_{H\subseteq N\setminus I} (-1)^{|H|} y_{I\cup H}
= \sum_{h=0}^{n-|I|} \binom{n-|I|}{h} (-1)^{h} \frac{\binom{\K}{|I|+h}}{\binom{n}{|I|+h}}
  \notag\\
&=  y_I \binom{\K-|I|-1}{n-|I|}  (-1)^{n-|I|}= (n+1) \binom{{\K}}{n+1} \frac{(-1)^{n-|I|}}{{\K}-|I|} \label{prob1}
\end{align}
where we use the identity $\sum_{\K=0}^m (-1)^{\K} \binom{n}{{\K}}=(-1)^m\binom{n-1}{m}$.

\end{document}